\pdfoutput=1
\documentclass[letterpaper,USenglish]{lipics-v2018}
\usepackage{graphicx}
\usepackage[noend]{algpseudocode}
\usepackage{algorithmicx}
\usepackage{algorithm}
\usepackage{cite}
\usepackage{hyperref}
\usepackage{thmtools,thm-restate}

\graphicspath{{figures/}}

\nolinenumbers

\title{Optimally Sorting Evolving Data}

\author{Juan Jose~Besa}{Dept. of Computer Science, Univ. of California, Irvine, Irvine, CA 92697 USA}{jjbesavi@uci.edu}{https://orcid.org/0000-0002-5676-7011}{}
\author{William E.~Devanny}{Dept. of Computer Science, Univ. of California, Irvine, Irvine, CA 92697 USA}{wdevanny@uci.edu}{}{Supported by an NSF Graduate Research Fellowship under grant DGE-1321846.}
\author{David Eppstein}{Dept. of Computer Science, Univ. of California, Irvine, Irvine, CA 92697 USA}{eppstein@uci.edu}{}{}
\author{Michael T.~Goodrich}{Dept. of Computer Science, Univ. of California, Irvine, Irvine, CA 92697 USA}{goodrich@uci.edu}{}{}
\author{Timothy Johnson}{Dept. of Computer Science, Univ. of California, Irvine, Irvine, CA 92697 USA}{tujohnso@uci.edu}{}{}
  
\authorrunning{J. Besa and W. Devanny and D. Eppstein and M. Goodrich and T. Johnson}
\Copyright{Juan Besa and William Devanny and David Eppstein and Michael Goodrich and Timothy Johnson}

\subjclass{}

\funding{This article reports on work supported by the
DARPA under agreement no.~AFRL FA8750-15-2-0092.
The views expressed are those of the authors and do not reflect the
official policy or position of the Department of Defense
or the U.S.~Government.
This work was also supported in part from NSF grants
1228639, 1526631,
1217322, 1618301, and 1616248.}

\EventEditors{Ioannis Chatzigiannakis, Christos Kaklamanis, D\'{a}niel Marx, and Don Sannella}
\EventNoEds{4}
\EventLongTitle{45th International Colloquium on Automata, Languages, and Programming (ICALP 2018)}
\EventShortTitle{ICALP 2018}
\EventAcronym{ICALP}
\EventYear{2018}
\EventDate{July 9--13, 2018}
\EventLocation{Prague, Czech Republic}
\EventLogo{eatcs}
\SeriesVolume{107}
\ArticleNo{194}

\begin{document}

\maketitle

\begin{abstract}
We give optimal sorting algorithms in the 
\emph{evolving data} framework, where an algorithm's input data
is changing while the algorithm is executing.
In this framework, instead of producing a final output, an algorithm 
attempts to maintain an output close to the correct output 
for the current state of the data, repeatedly updating its best estimate
of a correct output over time.
We show that a simple repeated insertion-sort algorithm
can maintain an $O(n)$ Kendall tau distance, with high probability,
between a maintained list 
and an an underlying total order of $n$ items 
in an evolving data model where
each comparison is followed by a swap between a random consecutive
pair of items in the underlying total order.
This result is asymptotically optimal, since there is an $\Omega(n)$ lower bound
for Kendall tau distance for this problem. 
Our result closes the gap between this lower bound
and the previous best algorithm for this problem, which maintains a Kendall tau
distance of $O(n\log \log n)$ with high probability. It also confirms previous experimental
results that suggested that insertion sort tends to perform better than quicksort in practice.
\end{abstract}

\section{Introduction}
In the classic version of the sorting problem,
we are given a set, $S$, of $n$ comparable items coming from a fixed total
order and asked to compute a permutation that places the items from
$S$ into non-decreasing order, and it is well-known
that this can be done using $O(n\log n)$ comparisons, which is
asymptotically optimal
(e.g., see~\cite{Cormen:2001,Goodrich:2014:ADA,knuth1998art}).
However, there are a number of interesting 
applications where this classic version of
the sorting problem doesn't apply.

For instance, consider the problem of maintaining a
ranking of a set of sports teams based on the results
of head-to-head matches.
A typical approach to this sorting problem is to assume there is 
a fixed underlying total order
for the teams, but that the outcomes of head-to-head matches 
(i.e., comparisons) are ``noisy'' in some way.
In this formulation, the ranking problem becomes a one-shot optimization
problem of finding the most-likely fixed total order given the outcomes
of the matches 
(e.g., see~\cite{Braverman:2008,Feige1994,Groz:2015,Hochbaum2006,Makarychev:2013}).
In this paper, we study an alternative, complementary motivating scenario,
however,
where instead of there being a fixed total order and noisy comparisons we
have a scenario where comparisons are accurate but the underlying
total order is evolving.
This scenario, for instance, 
captures the real-world phenomenon where sports teams
make mid-season changes to their player rosters and/or coaching staffs that result
in improved or degraded competitiveness relative to other teams.
That is, we are interested in the sorting problem for
\emph{evolving data}.

\subsection{Related Prior Work for Evolving Data}
Anagnostopoulos \textit{et al.}~\cite{sort11} introduce
the \emph{evolving data} framework, 
where an input data set is changing while an algorithm is processing it.
In this framework, instead of an algorithm taking a single input and producing 
a single output, an algorithm attempts to maintain an output close to the correct
output for the current state of the data, repeatedly updating its best estimate of
the correct output over time.
For instance,
Anagnostopoulos \textit{et al.}~\cite{sort11} mention the motivation of maintaining
an Internet ranking website that displays an
ordering of entities, such as political candidates, movies, or vacation spots,
based on evolving preferences.

Researchers have subsequently studied other 
interesting problems in the evolving data
framework, including 
the work of Kanade \textit{et al.}~\cite{kanade_et_al:LIPIcs:2016}
on stable matching with evolving preferences,
the work of Huang \textit{et al.}~\cite{Huang2017}
on selecting top-$k$ elements with evolving rankings,
the work of Zhang and Li~\cite{zhang2016shortest} on
shortest paths in evolving graphs,
the work of Anagnostopoulos \textit{et al.}~\cite{Anagnostopoulos:2012:AEG}
on st-connectivity and minimum spanning trees in evolving
graphs,
and the work of Bahmani \textit{et al.}~\cite{Bahmani:2012}
on PageRank in evolving graphs.
In each case, the goal is to maintain an output close to the correct one 
even as the underlying data is changing at a rate commensurate to the speed of the
algorithm.
By way of analogy, classical algorithms are to
evolving-data algorithms as throwing is to juggling.

\subsection{Problem Formulation for Sorting Evolving Data}
With respect to the
sorting problem for evolving data,
following the formulation of Anagnostopoulos \textit{et al.}~\cite{sort11},
we assume that we have a set, $S$, of $n$ distinct
items that are properly ordered
according to a total order relation, ``$<$''.
In any given time step,
we are allowed to compare any pair of items, $x$ and $y$, in $S$
according to the ``$<$'' relation and we learn the correct outcome of
this comparison.
After we perform such a comparison, $\alpha$ pairs 
of items that are 
currently consecutive according to the ``$<$'' relation are chosen
uniformly at random and their relative order is swapped.
As in previous work~\cite{sort11},
we focus on the case where $\alpha=1$, but
one can also consider versions of the problem where the ratio between
comparisons and random consecutive swaps is something other than one-to-one.
Still, this simplified version with a one-to-one ratio already raises some
interesting questions.

Since it is impossible in this scenario to maintain a list
that always reflects a strict ordering according to the ``$<$''
relation, our goal is to maintain a list with small 
\emph{Kendall tau} distance, which counts the number
of inversions, relative to the correct order.\footnote{Recall that an
  \emph{inversion}
  is a pair of items $u$ and $v$ such that $u$ comes before
  $v$ in a list but $u>v$.  An \emph{inversion} in 
  a permutation $\pi$ is a pair of elements $x\neq y$ 
  with $x < y$ and $\pi(x) > \pi(y)$.}
Anagnostopoulos \textit{et al.}~\cite{sort11} 
show that, for $\alpha=1$, the Kendall tau distance
between the maintained list and the underlying total order is 
$\Omega(n)$ in both expectation and with high probability.
They also show how to maintain this distance to be $O(n\log\log n)$,
with high probability, by performing a multiplexed 
batch of quicksort algorithms 
on small overlapping intervals of the list.
Recently, Besa Vial \textit{et al.} \cite{sort18} empirically show that 
repeated versions of quadratic-time algorithms
such as bubble sort and insertion sort seem to maintain an 
asymptotically optimal distance of $O(n)$. In fact, this linear upper bound seems to hold even if we allow $\alpha$, the number of random swaps at each step, to be a much larger constant.

\subsection{Our Contributions}
The main contribution of the present paper is to prove
that repeated insertion sort maintains
an asymptotically optimal Kendall tau distance, with high probability, for sorting
evolving data.
This algorithm repeatedly makes in-place insertion-sort passes 
(e.g., see~\cite{Cormen:2001,Goodrich:2014:ADA}) over the list,
$l_t$, maintained by our algorithm at each step $t$.
Each such pass moves the item at position $j$ to an earlier position in the 
list so long as it is bigger
than its predecessor in the list.
With each comparison done by this repeated insertion-sort 
algorithm, we assume that 
a consecutive pair of elements in the underlying ordered list, $l_t'$,
are chosen uniformly at random and swapped.
In spite of the uncertainty involved in sorting evolving data in this way,
we prove the following theorem, which is the main result of this paper.

\begin{theorem}\label{thm:ins-sort}
Running repeated insertion-sorts algorithm,
for every step $t =\Omega(n^2)$, the Kendall tau distance between 
the maintained list, $l_t$, and the underlying ordered list, $l_t'$,
is $O(n)$ with exponentially high probability.
\end{theorem}

That is, after an initialization period of $\Theta(n^2)$ steps, the repeated
insertion-sort algorithm converges to a steady state having an asymptotically 
optimal Kendall tau distance between the maintained list and the underlying
total order, with exponentially high probability.
We also show how to reduce this initialization period to be $\Theta(n\log n)$
steps, with high probability, by first performing a quicksort 
algorithm and then following that with the repeated insertion-sort
algorithm.

Intuitively, our proof of \autoref{thm:ins-sort} relies on two ideas:
the adaptivity of insertion sort and that, as time progresses, a constant
fraction of the random swaps fix inversions. Ignoring the random
swaps for now, when there are $k$ inversions, a complete 
execution of insertion sort performs
roughly $k+n$ comparisons and fixes the $k$ inversions
(e.g., see~\cite{Cormen:2001,Goodrich:2014:ADA}).
If an $\epsilon$
fraction of the random swaps fix inversions, then during insertion
sort $\epsilon(k+n)$ inversions are fixed by the random swaps and
$(1-\epsilon)(k+n)$ are introduced. Naively the total change in the
number of inversions is then $(1-2\epsilon)(k+n) -k$ and 
when $k > \frac{1-2\epsilon}{2\epsilon} n$, the number of inversions decreases.
So the number of inversions will decrease until $k = O(n)$. 

This simplistic intuition ignores two competing forces involved in
the heavy interplay between the random
swaps and insertion sort's runtime, however, in the evolving data model,
which necessarily complicates our proof. 
First, random swaps can cause an insertion-sort pass to end too early,
thereby causing insertion
sort to fix fewer inversions than normal. Second, as 
insertion sort progresses, it 
decreases the chance for a random swap to fix an inversion.
Analyzing these two interactions comprises the majority of our proof
of Theorem~\ref{thm:ins-sort}.

In Section~\ref{sec:ins-sort}, we present a complete proof of
Theorem~\ref{thm:ins-sort}. The most difficult component of
Theorem~\ref{thm:ins-sort}'s proof is Lemma~\ref{lem:many-invs}, which lower bounds the runtime of insertion sort in the evolving data model. The proof of Lemma~\ref{lem:many-invs} is presented separately in Section~\ref{sec:lem-proof}.

\section{Preliminaries}

The sorting algorithm we analyze in this paper for the evolving data
model is 
the repeated insertion-sort algorithm 
whose pseudocode is shown in \autoref{alg:rep-ins-sort}.

\begin{algorithm}[hbt]
\caption{Repeated insertion sort pseudocode}\label{alg:rep-ins-sort}
\small
\begin{algorithmic}
\Function{repeated\_insertion\_sort}{$l$}
  \While{true}
    \For{$i \gets 1$ to $n-1$}
      \State $j \gets i$
      \While{ $j > 0$ and $l[j] < l[j-1]$}
        \State swap $l[j]$ and $l[j-1]$
        \State $j \gets j-1$
        \EndWhile
    \EndFor
  \EndWhile
\EndFunction
\end{algorithmic}
\end{algorithm}

Formally, at time $t$, we denote the sorting algorithms' list as $l_t$
and we denote the underlying total order as $l'_t$. Together these two lists
define a permutation, $\sigma_t$, of the indices, where $\sigma_t(x) = y$
if the element at index $x$ in $l_t$ is at position $y$ in $l'_t$.
We define the \emph{simulated final state at time $t$} to be the state of $l$ 
obtained by freezing the current underlying total order, $l'_t$,
(i.e., no more random swaps)
and simulating the rest of the current round of insertion sort
(we refer to each iteration of the \textbf{while-true} loop in 
Algorithm~\ref{alg:rep-ins-sort} as a \emph{round}).
We then define a \emph{frozen-state} permutation,
$\hat{\sigma}_t$, where $\hat{\sigma}_t(x) = y$ if the element at index $x$ in 
the simulated final state at time $t$ as at index $y$ in $l'_t$.

Let us denote the number of inversions at time $t$, in
$\sigma_t$, with $I_t$. Throughout the paper, we may choose to drop
time subscripts if our meaning is clear. The Kendall tau distance
between two permutations $\pi_1$ and $\pi_2$ is the number of pairs
of elements $x\neq y$ such that $\pi_1(x) < \pi_1(y)$ 
and $\pi_2(x) > \pi_2(y)$. 
That is, the Kendall tau distance between $l_t$ and $l'_t$ is
equal to $I_t$, the number of inversions in $\sigma_t$.
Figure~\ref{fig:evol-model} shows the state of $l$, $l'$, $I$, and
$\sigma$ for two steps of an insertion sort (but not in the same round).

\begin{figure}[hbt]
\centering
\includegraphics[scale=1.0]{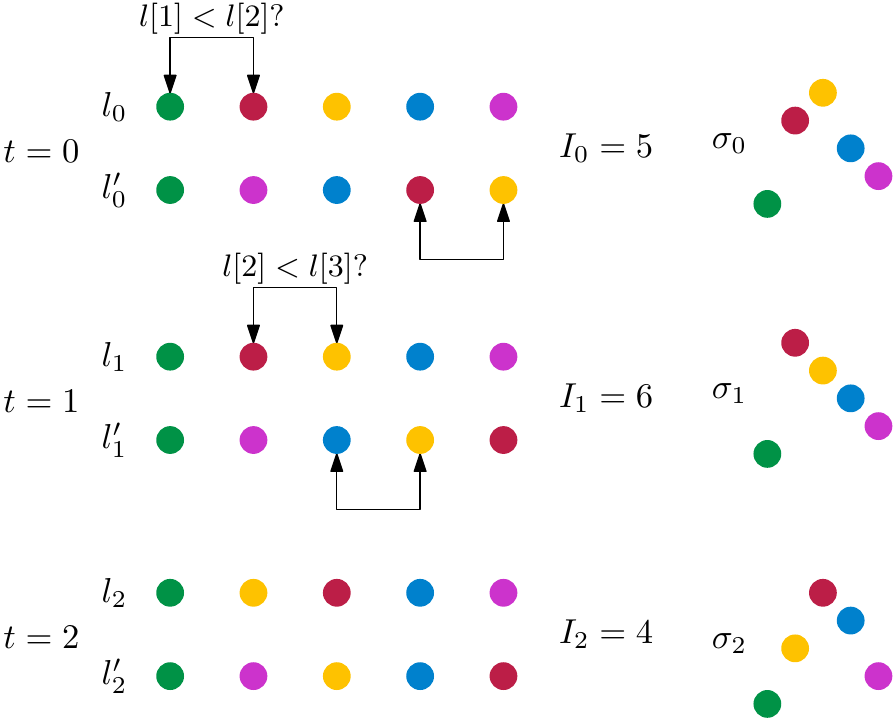}
\caption{Examples of $l$, $l'$, $I$, and $\sigma$ over two steps of an algorithm. In the first step the green and red elements are compared in $l$ and the red and yellow elements are swapped in $l'$. In the second step the red and yellow elements are compared and swapped in $l$ and the blue and yellow elements are swapped in $l'$.}\label{fig:evol-model}
\end{figure}

As the inner \textbf{while}-loop of Algorithm~\ref{alg:rep-ins-sort} executes,
we can view 
$l$ as being divided into three sets: the set containing just
the \emph{active} element, $l[j]$ (which we view as moving to
the left, starting from position $i$, 
as it is participating in comparisons and swaps), 
the \emph{semi-sorted} portion,
$l[0:i]$, not including $l[j]$, and the \emph{unsorted} portion, $l[i+1:n-1]$.
Note that if no random adjacent swaps were occurring in $l'$ (that
is, if we were executing insertion-sort in the classical algorithmic
model),
then the semi-sorted portion would be in sorted order.

We call the path from the root to the rightmost leaf of the Cartesian
tree the (right-to-left) minima path as the elements on this path are
the right-to-left minima in the list. The minima path is highlighted
in Figure~\ref{fig:cart-trees}. For a minimum, $l[k]$, denote with $M(k)$
the index of the element in the left subtree of $l[k]$ that maximizes
$\hat{\sigma}(k)$, i.e., the index of the largest element in the left subtree.

We use the phrase \emph{with high probability} to indicate when an event occurs with probability that tends towards $1$ as $n\rightarrow \infty$. When an event occurs with probability of the form $1-e^{-\mathop{poly}(n)}$, we say it occurs with \emph{exponentially high probability}.
During our analysis, we will make use of the following facts.

\begin{lemma}[Poisson approximation (Corollary 5.9 in \cite{alice})]
\label{thm:pois}
Let $X^{(m)}_1,\dots,X^{(m)}_n$ be the number of balls in each bin when $m$ balls are thrown uniformly at random into $n$ bins.
Let $Y^{(m)}_1,\dots,Y^{(m)}_n$ be independent Poisson random variables with $\lambda = m/n$.
Then for any event $\varepsilon(x_1,\dots,x_n)$:
\[ \Pr\left[\varepsilon\left(X^{(m)}_1,\dots,X^{(m)}_n\right)\right] \leq
e\sqrt{m} \Pr\left[\varepsilon\left(Y^{(m)}_1,\dots,Y^{(m)}_n\right)\right]. \]
\end{lemma}

\begin{lemma}[Hoeffding's inequality (Theorem 2 in \cite{hoeffding})]
\label{thm:hoeffding}
If $X_1,\dots,X_n$ are independent random variables and $a_k \leq X_k \leq b_k$ for $k=1,\dots,n$, then for $t>0$:
\[\Pr\left[\sum_k X_k - E\left[ \sum_k X_k \right] \geq tn\right] \leq
e^{-2n^2t^2/\left(\sum_k \left(b_k - a_k\right)^2\right)}. \]
\end{lemma}

\section{Sorting Evolving Data with Repeated Insertion Sort}\label{sec:ins-sort}
Let us begin with some simple bounds with respect to a single round of 
insertion sort.

\begin{lemma}\label{lem:ins-basics}
If a round of insertion sort starts at time $t_s$ and finishes at time $t_e$, then
\begin{enumerate}
\item $t_e - t_s = F + n-1$, where $F$ is the number of inversions fixed 
(at the time of a comparison in the inner \textbf{while}-loop) by 
this round of insertion sort.
\label{en-1}
\item
$t_e-t_s< n^2/2$ 
\label{en-1b}
\item for any $t_s \leq t \leq t_e$, $I_t - I_{t_s} < n$.
\label{en-2}
\end{enumerate}
\end{lemma}
\begin{proof}
(\ref{en-1}):
For each iteration of the outer \textbf{for}-loop,
each comparison in the inner \textbf{while}-loop either fixes an inversion
(at the time of that comparison)
or fails to fix an inversion and completes the inner 
\textbf{while}-loop.
Note that this ``failed'' comparison may not have compared elements of
$l$, but may have short circuited due to $j\leq 0$. 
Nevertheless, every comparison that doesn't fail fixes an inversion (at the time
of that comparison); hence, each non-failing comparison is counted in $F$.

(\ref{en-1b}):
In any round, there are at most $n(n-1)/2$ 
comparisons, by the formulations of the outer \textbf{for}-loop and inner
\textbf{while}-loop.

(\ref{en-2}):
At time $t$, the round of insertion sort will have executed $t-t_s$ steps. Of those steps, at least $t - t_s - (n-1)$ comparisons resulted in a swap that removed an inversion and at most $n-1$ comparisons did not result in a change to $l$. The random swaps occurring during these comparisons introduced at most $t - t_s$ inversions. So $I_t - I_{t_s} \leq t - t_s - \bigl(t - t_s - (n-1)\bigr) = n-1$. 
\end{proof}

We next assert the following two lemmas, which are used in the next section and proved later.

\begin{restatable}{lemma}{swapproblemma}
\label{lem:swap-prob}
There exists a constant, $0<\epsilon < 1$, such that, for a round of insertion sort that takes time $t^*$, at least $\epsilon t^*$ of the random adjacent swaps in $l'$ 
decrease $I$ during the round, with exponentially high probability.
\end{restatable}
\begin{proof}
See Appendix~\ref{app:swap-prob}.
\end{proof}

\begin{lemma}\label{lem:many-invs}
If a round of insertion sort starts at time $t_s$ with $I_{t_s} \geq
(12c^2 +2c)n$ and finishes at time $t_e$, then,
with exponentially high probability, $t_e-t_s \geq cn$,
i.e., the insertion sort round takes at least $cn$ steps.
\end{lemma}
\begin{proof}
See Section~\ref{sec:lem-proof}.
\end{proof}

\subsection{Proof of Theorem~\ref{thm:ins-sort}}
Armed with the above lemmas (albeit postponing the proofs of Lemma~\ref{lem:swap-prob} and Lemma~\ref{lem:many-invs}), let us prove our main theorem.

\medskip \noindent
\textbf{Theorem~\ref{thm:ins-sort}.}
\textit{There exists a constant, $0<\epsilon<1$, such that,
when running the repeated insertion-sort algorithm,
for every step $t > (1+1/\epsilon)n^2$, the Kendall tau distance between
the maintained list, $l_t$, and the underlying ordered list, $l_t'$,
is $O(n)$, with exponentially high probability.}

\begin{proof}
By Lemma~\ref{lem:swap-prob}, there exists a constant $0<\epsilon<1$ such that at
least an $\epsilon$ fraction of all of the random swaps during a
round of insertion sort fix inversions. 
Consider an epoch of the last $(1+1/\epsilon)n^2$ steps of 
the repeated insertion-sort algorithm, that is, from
time $t'=t-(1+1/\epsilon)n^2$ to $t$.
During this epoch, 
some number, $m\ge 1$, of complete rounds of insertion sort are performed from start to end (by Lemma~\ref{lem:ins-basics}). 
Denote with $t_k$ the time at which insertion-sort 
round $k$ ends (and round $k+1$ begins),
and let $t_m$ denote the end time of the final complete round,
during this epoch.
By construction, observe that $t'\le t_0$ and $t_m \leq t$. 
Furthermore, because the insertion-sort rounds running before $t_0$ and 
after $t_m$ take fewer than $n^2/2$ steps (by Lemma~\ref{lem:ins-basics}), 
$t_m - t_0\,\ge\, {n^2}/{\epsilon}$.

The remainder of the proof consists
of two parts. In the first part, we show that for some complete round of
insertion sort ending at time $t_k\le t$,
$I_{t_k}$ is $O(n)$,
with exponentially high probability.
In the second part, we show that once we achieve
$I_{t_k}$ being $O(n)$, for $t_k\le t$,
then $I_{t}$ is $O(n)$, with exponentially high probability.

For the first part,
suppose, for the sake of a contradiction,
$I_{t_k} > \bigl(12(\frac{1}{\epsilon})^2 + \frac{2}{\epsilon}\bigr)n$, for all $0\leq k \leq m$. 
Then, by a union bound over the polynomial number of rounds, Lemma~\ref{lem:many-invs} applies to every such round of insertion sort. 
So, with exponentially high probability, each round takes at least $n/\epsilon$ 
steps. Moreover, by Lemma~\ref{lem:swap-prob},
with exponential probability, an $\epsilon$ fraction of
the random swaps from $t_m$ to $t_0$
will decrease the number of inversions. That is, these random swaps
increase the number of inversions by at most
\[
(1-\epsilon)(t_m-t_0)-\epsilon(t_m-t_0)=(1-2\epsilon)(t_m-t_0),
\]
with exponentially high probability.
Furthermore, by Lemma~\ref{lem:ins-basics}, 
at least a $\frac{(1/\epsilon)-1}{1/\epsilon}=1-\epsilon$ fraction of 
the insertion-sort steps fix inversions
(at the time of a comparison).
Therefore, with exponentially high probability,
we have the following:
\begin{align*}
I_{t_m} & \leq I_{t_0} - (1-\epsilon)(t_m-t_0) + (1-2\epsilon) (t_m-t_0) \\
 & = I_{t_0} - \epsilon(t_m-t_0) \\
 & \leq I_{t_0} - n^2.
\end{align*}
But, since $I_{t_0} < n^2$, the above bound implies that 
$I_{t_m} < 0$, which is a contradiction.
Therefore, with exponentially high probability,
there is a $k\le m$ 
such that $I_{t_k} \leq (12(\frac{1}{\epsilon})^2 + \frac{2}{\epsilon})n$. 

For the second part,
we show that the probability for a round $\ell>k$ to 
have $I_{t_{\ell}} > (12(\frac{1}{\epsilon})^2 + \frac{2}{\epsilon} + 1)n$ 
is exponentially small, by considering two cases (and their implied union-bound
argument):
\begin{itemize}
\item If $I_{t_{\ell-1}} \leq (12(\frac{1}{\epsilon})^2 + \frac{2}{\epsilon})n$, 
then Lemma~\ref{lem:ins-basics} implies $I_{t_{\ell}} \leq (12(\frac{1}{\epsilon})^2 + \frac{2}{\epsilon} + 1)n$.
\item If $(12(\frac{1}{\epsilon})^2 + \frac{2}{\epsilon})n \leq 
I_{t_{\ell-1}} \leq (12(\frac{1}{\epsilon})^2 + \frac{2}{\epsilon} + 1)n$, 
then, similar to the argument given above,
during a round of insertion sort, $\ell$, at least a $1-\epsilon$ fraction of the steps fix an inversion, and an $\epsilon$ fraction of the steps do nothing. Also at least an $\epsilon$ fraction of the random swaps fix inversions, while a $1-\epsilon$ fraction add inversions. Finally, the total length of the round is $t_{\ell} - t_{\ell-1}$.
Thus, with exponentially high probability, the total change in inversions is at 
most $-\epsilon(t_{\ell} - t_{\ell-1})$ and $I_{t_{\ell}} < I_{t_{\ell-1}}$.
\end{itemize}
Therefore, by a union bound over the polynomial number of insertion-sort
rounds, the probability that any $I_{t_\ell} > (12(\frac{1}{\epsilon})^2 + \frac{2}{\epsilon} + 1)n$ for $k < \ell \leq m$ is exponentially small. 
By Lemma~\ref{lem:ins-basics}, $I_t \leq I_{t_m} + n$. So, with exponentially high probability, $I_{t_m} \leq (12(\frac{1}{\epsilon})^2 + \frac{2}{\epsilon} + 1)n = O(n)$ and $I_t = O(n)$, completing the proof.
\end{proof}

\subsection{Improved Convergence Rate}

In this subsection, we provide an algorithm that converges to $O(n)$ inversions more quickly. 
To achieve the steady state of $O(n)$ inversions, 
repeated insertion sort performs $\Theta(n^2)$ comparisons.
But this running time to reach a steady state is a worst-case based on 
the fact that the running time of insertion sort is $O(n+I)$, where
$I$ is the number of initial inversions in the list, and, in the
worst case, $I$ is $\Theta(n^2)$.
By simply running a round of quicksort on $l$ first, we can achieve 
a steady state of $O(n)$ inversions after just $\Theta(n\log n)$ comparisons.
See Algorithm~\ref{alg:quick-ins-sort}.
That is, we have the following.

\begin{algorithm}[hbt]
\caption{Quicksort followed by repeated insertion sort pseudocode}\label{alg:quick-ins-sort}
\begin{algorithmic}
\Function{quick\_then\_insertion\_sort}{$l$}
  \State quicksort($l$)
  \While{true}
    \For{$i \gets 1$ to $n-1$}
      \State $j \gets i$
      \While{ $j > 0$ and $l[j] < l[j-1]$}
        \State swap $l[j]$ and $l[j-1]$
        \State $j \gets j-1$
        \EndWhile
    \EndFor
  \EndWhile
\EndFunction
\end{algorithmic}
\end{algorithm}

\begin{theorem}\label{thm:faster-conv}
When running Algorithm~\ref{alg:quick-ins-sort}, for every $t = \Omega(n\log n)$, $I_t$ is $O(n)$ with high probability.
\end{theorem}
\begin{proof}
By the results of Anagnostopoulos {\it et al.}~\cite{sort11},
the initial round of quicksort takes $\Theta(n\log n)$ comparisons 
and afterwards the number of inversions (that is, the Kendall tau
distance between the maintained list and the true total order) is
$O(n\log n)$, with high probability.
Using a nearly identical argument to the proof of Theorem~\ref{thm:ins-sort}, 
and the fact that an insertion-sort round takes $O(I+n)$ time to resolve
$I$ inversions,
the repeated insertion-sort algorithm will, with high probability, achieve $O(n)$ inversions in an additional $O(n\log n)$ steps.
From that point on, it will maintain a Kendall tau distance of
$O(n)$, with high probability.
\end{proof}

\section{Proof of Lemma~\ref{lem:many-invs}}\label{sec:lem-proof}

Recall~Lemma~\ref{lem:many-invs}, which establishes a lower bound for the running
time of an insertion-sort round, given a sufficiently large amount of inversions 
relative to the underlying total order.

\medskip\noindent
\textbf{Lemma~\ref{lem:many-invs}.}
\textit{If a round of insertion sort starts at time $t_s$ with $I_{t_s} \geq
(12c^2 +2c)n$ and finishes at time $t_e$, then,
with exponentially high probability, $t_e-t_s \geq cn$,
i.e., the insertion sort round takes at least $cn$ steps.}

\medskip
The main difficulty in proving Lemma~\ref{lem:many-invs} is
understanding how the adjacent random swaps in $l'$ affect the
runtime of the current round of insertion sort on $l$. Let $S_t$ be
the number of steps left to perform in the current round of insertion
sort if there were no more random adjacent swaps in $l'$. In essence,
$S$ can be thought of as an estimate of the remaining time in the
current insertion sort round. If a new round of insertion sort is
started at time $t_s$, then $S_{t_s-1} = 1$ and $I_{t_s}\leq  S_{t_s} \leq I_{t_s} + n-1$. 
Each step of an insertion sort round decreases $S$ by one and the following random swap may increase or decrease $S$ by some amount. Figure~\ref{fig:bad-swaps} illustrates an example where one random adjacent swap in $l'$ decreases $S$ by a non-constant amount (relative to $n$).

\begin{figure}[hbt]
\centering
\includegraphics{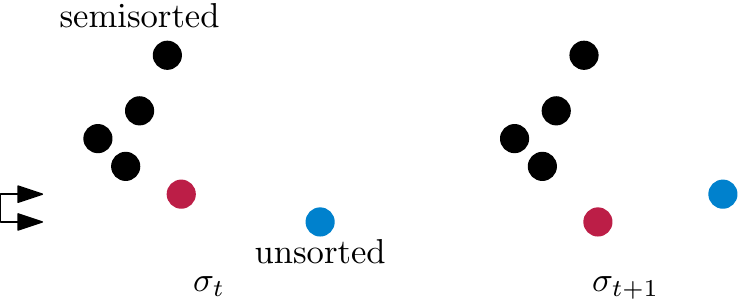}
\caption{An example where swapping the ordering of the red and blue elements in $l'$ creates multiple blocked inversions between the blue element and the black elements. Recall that our list is partitioned into the semisorted region, which contains elements that have already been compared in this round, and the unsorted region.}\label{fig:bad-swaps}
\end{figure}

A random adjacent swap in $l'$ involving two elements in the unsorted portion of $l$ will either increase or decrease $S$ by one depending on whether it introduces or removes an inversion. Random adjacent swaps involving elements in the semi-sorted portion have more complex effects on $S$.

An inversion currently in the list $\bigl(l[a],l[b]\bigr)$ will be fixed by insertion sort if $l[a]$ and $l[b]$ will be compared and the two are swapped. Because $a < b$, $l[b]$ must be the active element during this comparison.
An inversion $\bigl(l[a],l[b]\bigr)$ will not be fixed by insertion sort if $l[b]$ was already inserted into the semi-sorted portion or there is some element $l[c]$ in the semi-sorted portion with $a < c < b$ and $\sigma(c) < \sigma(b)$.
We call an inversion with $l[b]$ in the semi-sorted portion a \emph{stuck} inversion and an inversion with a smaller semi-sorted element between the pair a \emph{blocked} inversion.
We say an element $l[c]$ in the semi-sorted portion of $l$ \emph{blocks} an inversion $\bigl(l[a],l[b]\bigr)$ with $a \leq i$ and $l[b]$ either the active element or in the unsorted portion of $l$, if $l[c]$ is in the semi-sorted portion of $l$ with $a < c < b$ and $\sigma(c) < \sigma(b)$. Note that there may be multiple elements blocking a particular inversion. Figure~\ref{fig:blocked-and-stuck} shows examples of these two types of inversions.

\begin{figure}[hbt]
\centering
\includegraphics{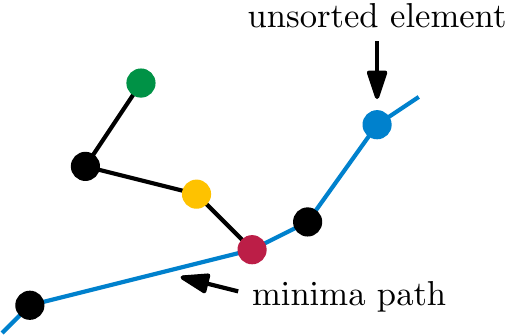}
\caption{In this Cartesian tree (see Appendix~\ref{app:swap-counters}), the green-blue pair is a blocked inversion 
and the green-yellow pair is a stuck inversion. Both pairs of inversions blame the red element.}\label{fig:blocked-and-stuck}
\end{figure}

We denote the number of ``bad'' inversions at time $t$ that will not be fixed with $B_t$. 
That is, $B_t$ is the sum of the blocked and stuck inversions.
At the end of an insertion-sort round every inversion present at the
start was either fixed by the insertion sort, fixed by a random
adjacent swap in $l'$, or is currently stuck. 
No elements can be blocked at the end of an insertion-sort round,
because the semi-sorted portion is the entire list. Stuck inversions are either created by random adjacent swaps in $l'$ or were blocked inversions and insertion sort finished inserting the right element of the pair. Blocked inversions are only introduced by the random adjacent swaps in $l'$. Thus $B_t$ is unaffected by the steps of insertion sort. 

Every inversion present at the start must be fixed by a step of
insertion sort, be fixed by a random swap, or it will end up ``bad''.
Therefore, for any given time, $t$, by using naive upper bounds based on
the facts that every insertion sort step can fix an inversion and 
every random adjacent swap can remove an inversion, we can immediately
derive the following:

\begin{lemma}\label{lem:est-bound}
For an insertion sort round that starts at time $t_s$ and ends at time $t_e$, if $t_s \leq t \leq t_e$, then $S_t \geq I_{t_s} - 2(t-t_s) - B_t$.
\end{lemma}

Since, when an insertion sort round finishes, $S_{t_e-1} = 1$, Lemma~\ref{lem:est-bound} implies $2(t_e-t_s-1) + B_{t_e} +1 \geq I_{t_s}$. If we understand how $B$ changes with each random adjacent swap in $l'$, then we can bound how long insertion sort needs to run for this inequality to be true.

We associate the blocked and stuck inversions with elements that we say are \emph{blamed} for the inversions. A blocked inversion $\bigl(l[a],l[b]\bigr)$ blames the element $l[c]$ with $a < c < b$ and minimum $\sigma(c)$. Note that $l[c]$ is on the minima path of the modified Cartesian tree (see Appendix~\ref{app:swap-counters}), and $l[a]$ is in the left subtree of $l[c]$. A stuck inversion either blames the element on the minima path whose subtree contains both $l[a]$ and $l[b]$ or if they appear in different subtrees, the inversion blames the element $l[c]$ with $a<c<b$ and minimum $\sigma(c)$. Again note that the blamed element is on the minima path and $l[a]$ is in the blamed element's left subtree. The bad inversions in Figure~\ref{fig:blocked-and-stuck} blame the red element.

Whether stuck or blocked, every inversion blames an element on the
minima path and the left element of the inverted pair appears in that
minimum's subtree. If $l[k]$ is on the minima path, $M(k)$ is the
index of the element in $l[k]$'s subtree with maximum $\sigma(M(k))$,
and an inversion $\bigl(l[a],l[b]\bigr)$ has $l[a]$ in $l[k]$'s subtree, then
both $l[a]$ and $l[b]$ are in the range $\sigma(k)$ to
$\sigma(M(k))$. So we can upper bound $B_t$ by
$\sum_{k=0}^{n-1} (\sigma(M(k)) - \sigma(k))^2$,
where we extend $M$ to non-minima indices with $M(k) = k$ if $k$ is not the index of a minima in $l$. 

\subsection{Bounding the Number of Blocked and Stuck Inversions with Counters}

For the purposes of bounding $B_t$, we conceptually associate two
counters, $\textit{Inc}(x)$ and $\textit{Dec}(x)$, with each element, $x$. The counters are initialized to zero at the start of an insertion sort round. When an element $x$ is increased by a random swap in $l'$, we increment $\textit{Inc}(x)$ and when $x$ is decreased by a random swap in $l'$, we increment $\textit{Dec}(x)$. After the random swap occurs, we may choose to exchange some of the counters between pairs of elements, but we will always maintain the following invariant:

\textbf{Invariant 1.} 
\textit{For an element, $l[k]$, on the minima path,} 
 \[\textit{Inc}\bigl(l[M(k)]\bigr) + \textit{Dec}\bigl(l[k]\bigr) \geq \sigma\bigl(M(k)\bigr) - \sigma(k).\]

This invariant allows us to prove the following Lemma:
\begin{lemma}\label{lem:counter-sums}
If $\sum_{k=0}^{n-1} \textit{Inc}\bigl(l[k]\bigr)^2 < \kappa$ and $\sum_{k=0}^{n-1} \textit{Dec}\bigl(l[k]\bigr)^2 < \kappa$, then $B_t \leq 4\kappa$.
\end{lemma}
\begin{proof}
\begin{align}
B_t & \leq \sum_{k=0}^{n-1} \Bigl(\sigma(M(k)\bigr) - \sigma(k)\Bigr)^2 & \nonumber \\
 & \leq \sum_{k=0}^{n-1} \Bigl(\textit{Inc}\bigl(M(k)\bigr) + \textit{Dec}(k)\Bigr)^2 & \text{By Invariant 1} \label{eq:vector-sum}
\end{align}

By the assumptions of this lemma, interpreting $\textit{Inc}$ and $\textit{Dec}$ as two $n$-dimensional vectors, we know their lengths are both less than $\sqrt{\kappa}$. Equation~\ref{eq:vector-sum} is the squared length of the sum of the $\textit{Dec}$ and $\textit{Inc}$ vectors with the entries of $\textit{Inc}$ permuted by the function $M$. By the triangle inequality, the length of their sum is at most $2\sqrt{\kappa}$ and so the squared length of their sum is at most $4\kappa$.

Therefore, $B_t \leq 4\kappa$.
\end{proof}

In the appendix we prove the following lemma for these increment and decrement counters.
\begin{restatable}{lemma}{incdeccounterlemma}
There is a counter maintenance strategy that maintains Invariant 1 such that after each random adjacent swap in $l'$, the corresponding counters are incremented and then some counters are exchanged between pairs of elements.
\end{restatable}

\subsection{Bounding the Counters with Balls and Bins}

We model the $\textit{Inc}$ and $\textit{Dec}$ counters each with a balls and bins process and analyze the sum of squares of balls in each bin. Each element in $l$ is associated with one of $n$ bins. When an element's $\textit{Inc}$ counter is increased, throw a ball into the corresponding bin. If a pair of $\textit{Inc}$ counters are exchanged, exchange the set of balls in the two corresponding bins. The $\textit{Dec}$ counters can be modeled similarly. 

This process is almost identical to throwing balls into $n$ bins uniformly at random. 
Note that the exchanging of balls in pairs of bins takes place after
a ball has been placed in a chosen bin, effectively permuting two bin
labels in between steps. If every bin was equally likely to be hit at
each time step, then permuting the bin labels in this way would not
change the final sum of squares and the exchanging of counters could
be ignored entirely. Unfortunately the bin for the element at
$l[n-1]$ in the case of $\textit{Inc}$ counters or $l[0]$ in the case of $\textit{Dec}$
counters cannot be hit, i.e., there is a forbidden bin controlled by the counter swapping strategy.
However, even when in each round the forbidden bin is adversarially chosen, the sum of squares of the number of balls in each bin will be stochastically dominated by a strategy of always forbidding the bin with the lowest number of balls. 
Therefore, the sum of squares of $m$ balls being thrown uniformly at random into $n-1$ bins stochastically dominates the sum of squares of the $\textit{Inc}$ (or $\textit{Dec}$) counters after $m$ steps.

\begin{theorem}\label{thm:balls-bins-squared}
If $cn$ balls are each thrown uniformly at random into $n$ bins with $c > e$, then the sum over the bins of the square of the number of balls in each bin is at most $3c^2n$ with exponentially high probability.
\end{theorem}
\begin{proof}
Let $X_1,\dots,X_n$ be random variables where $X_k$ is the number of balls in bin $k$ and let $Y_1,\dots,Y_n$ be independent Poisson random variables with $\lambda = c$. 

By the Poisson approximation, Lemma~\ref{thm:pois}, 
\[ \Pr\left[\sum_k X_k^2 \geq 3c^2n\right] \leq e\sqrt{cn}
\Pr\left[\sum_k Y_k^2 \geq 3c^2n\right].\]

Let $Z_k$ be the event that $Y_k \geq ecn^{1/6}$ and $Z$ be the event that at least one $Z_k$ occurs.

\begin{align*}
\Pr[Z] & \leq n\Pr[Z_1] \quad \text{by a union bound.}\\
\Pr[Z_1] & = e^{-c} \sum_{k=ecn^{1/6}}^\infty \frac{c^k}{k!} \leq e^{-c} \sum_{k=ecn^{1/6}}^\infty \frac{c^k}{e\left(\frac{k}{e}\right)^{k}}\\
 & = e^{-c-1} \sum_{k=ecn^{1/6}}^\infty \left(\frac{ec}{k}\right)^k \leq e^{-c-1} \sum_{k=ecn^{1/6}}^\infty \left(\frac{1}{n^{1/6}}\right)^k\\
 & = e^{-c-1} (n^{1/6})^{-ecn^{1/6}} \sum_{k=0}^\infty
 \frac{1}{n^{1/6}}^k \leq e^{-c}n^{-\frac{ec}{6}n^{1/6}}.\\
 \Rightarrow \Pr[Z] & \leq \frac{n}{e^cn^{\frac{ec}{6}n^{1/6}}} \leq
 e^{-\Omega(n^{1/6})}.
\end{align*}

Letting $Y = \sum_k Y_k^2$:
\begin{align*}
E[Y | \neg Z] \leq E[Y]  = nE[Y_1^2] = n\left(c + c^2\right) \leq 2c^2n.
\end{align*}

Given $\neg Z$, $(Y_k)^2 \in [0, ecn^{1/3}]$. So we can apply
Hoeffding's inequality, Lemma~\ref{thm:hoeffding}, to get:
\[\Pr\left[Y - E\left[Y|\neg Z\right] \geq tn | \neg Z\right] \leq
e^{-2t^2n^2/\left( n \left(ecn^{1/3}\right)^2\right)}. \]

Setting $t = c^2$, we have:
\begin{align*}
\Pr\left[Y - E\left[Y|\neg Z\right] \geq c^2n | \neg Z\right] &\leq e^{\left(-2c^4n^2\right)/\left(n\left(ecn^{1/3}\right)^2\right)}\\
 & \leq e^{-2n^{1/3}}.
\end{align*}

Because $E\left[Y|\neg Z\right] \leq 2c^2n$, we have $\Pr[Y \geq 3c^2n | \neg Z] \leq e^{-\Omega(n^{1/3})}$.
\begin{align*}
\Pr\left[Y \geq 3 c^2n \right] & = \Pr\left[Y \leq 3 c^2n\text{
and } Z\right] + \Pr\left[Y \leq c^2n \text{ and } \neg Z\right]\\
 & \leq \Pr[Z] + \Pr\left[Y \leq 3c^2 n | \neg Z\right]\\
 & \leq e^{-\Omega(n^{1/6})} + \Pr[Y - E[Y|\neg Z] \geq c^2n | \neg Z]\\
 & \leq e^{-\Omega(n^{1/6})} + e^{-\Omega(n^{1/3})} \leq
 2e^{-\Omega(n^{1/6})}.
\end{align*}

Thus, we can conclude $\Pr[\sum_k X_k^2 \geq 3c^2n] \leq \frac{2e\sqrt{cn}}{e^{\Omega(n^{1/6})}} \leq e^{-\mathop{poly}(n)}$.
\end{proof}

Recall that by Lemma~\ref{lem:est-bound}, if an insertion-sort round
ends at time $t$, then $I_{t_s} \leq 2(t-t_s) + B_{t} + 1$.
Theorem~\ref{thm:balls-bins-squared} and a simple union bound tell us
that if $t \leq t_s+cn$, then $\sum_{k=0}^{n-1} \textit{Inc}\bigl(l[k]\bigr)^2 \leq
3c^2(n-1)$ and $\sum_{k=0}^{n-1} \textit{Dec}\bigl(l[k]\bigr)^2 \leq 3c^2(n-1)$ with
exponentially high probability. So by Lemma~\ref{lem:counter-sums}, $B_t \leq 12c^2n$. 

Recall that when the insertion sort round finishes, $2(t_e-t_s-1) + B_{t_e} +1 \geq I_{t_s}$. 
If fewer than $cn$ steps have been performed, the left hand side of this inequality is less than $(12c^2 + 2c)n$ with exponentially high probability.
Therefore, if we started with $(12c^2 + 2c)n$ inversions, the current
round of insertion sort must perform at least $cn$ steps with
exponentially high probability; otherwise, there are unfixed but still
``good'' inversions. This completes the proof of Lemma~\ref{lem:many-invs}.

\section{Conclusion}

We have shown that,
although it is much simpler than quicksort and only fixes at most 
one inversion in each step, repeated insertion sort leads to the 
asymptotically optimal number of inversions in the evolving data model. 
We have also shown that by using a single round of quicksort before our repeated insertion sort,
we can get to this steady state after an initial phase of $O(n\log n)$ steps, which is also asymptotically optimal.

For future work, it would be interesting to explore whether our results can
be composed with other problems involving algorithms for evolving data, where
sorting is a subcomponent.
In addition, our analysis in this paper is specific to insertion sort, and only applies when exactly one random swap is performed after each comparison. We would like to extend this to other sorting algorithms that have been shown to perform well in practice and to the case in which the number of random swaps per comparison is a larger constant.
Finally, it would also be interesting to explore whether one can derive a much better
$\epsilon$ value than we derived in the proof of Lemma~\ref{lem:swap-prob}.

\bibliographystyle{plainurl}
\bibliography{refs}

\clearpage
\begin{appendix}
\section{Proof of Lemma 4} \label{app:swap-prob}
\swapproblemma*

\begin{proof}
We call a random adjacent swap that decreases the number 
of inversions, $I$, during the insertion-sort round a \emph{good swap}.

Break the time interval for this round of insertion sort into epochs, each of size between $n/32$ and $n/16$  (this is possible because $t^* \geq n-1$, by Lemma~\ref{lem:ins-basics}) and let $t_k$ be the start of epoch $k$. Denote the length of epoch $k$ by $t^*_k = t_{k+1} - t_k$. Given the values of $i$ and $j$ at $t_k$, only the elements in the ranges $l[j-n/16,j]$ and $l[i-n/16,i+n/16]$ will be involved in insertion sort comparisons during epoch $k$. This set of potentially compared elements has size at most $3n/16$. 

Consider the set of adjacent disjoint 4-tuples in $l'$, $l'[4a],l'[4a+1],l'[4a+2],l'[4a+3]$ for $a = 0,1,\dots,n/4$. There are $n/4$ of these tuples and so there are at least $n/4 - 3n/16 = n/16$ tuples whose elements cannot be involved in comparisons during a given epoch. Call such a tuple of elements an \emph{untouchable tuple}.

We now examine just the swaps during one specific epoch. Let $X_i$ be the number of random adjacent swaps that swap $l'[i]$ with $l'[i+1]$ for $i=0,1,\dots,n-1$. Let $Y_i$ be independent identically distributed Poisson random variables with parameter $\lambda = \frac{t^*_k}{n-1}$ for $i=0,1,\dots,n-1$. Note that $1/32 \leq \lambda \leq \frac{n}{16(n-1)} \leq 1/15$ for large enough $n$.
Let $f(z_1,z_2,\dots,z_n)$ be the function that counts how many $a = 0,1,\dots,n/4$ there are such that the tuple $l'[4a],l'[4a+1],l'[4a+2],l'[4a+3]$ is untouchable and $z_{4a} = 0$, $z_{4a+1} = 2$, and $z_{4a+2} = 0$.

By the Poisson approximation, Lemma~\ref{thm:pois}, for any $\delta > 0$, 
\[\Pr\bigl[f(X_1,X_2,\dots,X_{n-1}) \leq \delta n\bigr] \leq
e\sqrt{n/16}\, \Pr\bigl[f(Y_1,Y_2,\dots,Y_{n-1}) \leq \delta n\bigr].\]

As previously stated, there are at least $n/16$ untouchable tuples. Because the $Y_i$ are independent, for an untouchable 4-tuple $l'[4a],l'[4a+1],l'[4a+2],l'[4a+3]$,
\begin{align*}
\Pr\bigl[Y_{4a} = 0,Y_{4a+1} = 2,Y_{4a+2} = 0\bigr] &= \frac{e^{-3\lambda} \lambda^2}{0!2!0!}\\
 & \geq \frac{e^{-3/15}\left(1/32\right)^2}{2} \\
  & \geq \frac{3}{10,000}
\end{align*}
$f(Y_1,Y_2,\dots,Y_{n-1})$ is the sum of at least $n/16$ independent
indicator random variables that each have probability at least
$3/10,000$ of being $1$. Thus $E[f(Y_1,Y_2,\dots,Y_{n-1})] \geq
\frac{3n}{160,000}$. Therefore, by a Chernoff bound from \cite{alice}:

\begin{align*}
\Pr\left[f\left(Y_1,Y_2,\dots,Y_{n-1}\right) \leq \left(1-\frac{1}{2}\right) \frac{3n}{160,000}\right] &\leq e^{-\Omega(n)}\\
\Pr\left[f\left(X_1,X_2,\dots,X_{n-1}\right) \leq \left(1-\frac{1}{2}\right) \frac{3n}{160,000}\right] &\leq \frac{e\sqrt{n/16}}{e^{\Omega(n)}} \leq e^{-\mathop{poly}(n)}
\end{align*}

Therefore, within each epoch of the insertion sort round there are at
least $\frac{3}{320,000} n$ untouchable tuples where the middle pair
of indices are swapped twice and the other two pairs are not swapped,
with exponentially high probability. In each of these tuples one of the two swaps must have been a good swap.

So we can conclude that for each epoch, with exponentially high
probability, there are $\frac{3}{320,000} n$ good swaps. Because there are at least $\frac{t^*}{n/16}$ epochs, setting $\epsilon = \frac{3}{20,000}$ 
implies there are at least $\epsilon t^*$ good swaps during the
entire insertion sort round, with exponentially high probability.
\end{proof}

\section{Counter swapping} \label{app:swap-counters}
Given a list, $L$, of $m$ numbers with no two equal numbers, the
\emph{Cartesian tree}~\cite{Vuillemin:1980} of $L$ is a binary rooted tree on the numbers where the root is the minimum element $L[k]$, the left subtree of the root is the Cartesian tree of $L[0:k-1]$, and the right subtree of the root is the Cartesian tree of $L[k+1:m]$.
In our analysis, we will primarily consider the Cartesian tree of the simulated final state at time $t$ where $L[k] = \hat{\sigma}_t(k)$ in the frozen-state permutation
$\hat{\sigma}_t$.
We also choose to include two additional elements, $L[-1] = -1$ and $L[n] = n$,
for boundary cases.
Figure~\ref{fig:cart-trees} shows an example Cartesian tree we might consider.
The Cartesian trees we consider are only for the sake of analysis. They are not explicitly constructed.
 
\begin{figure}[!b]
\centering
\includegraphics[scale=1.1]{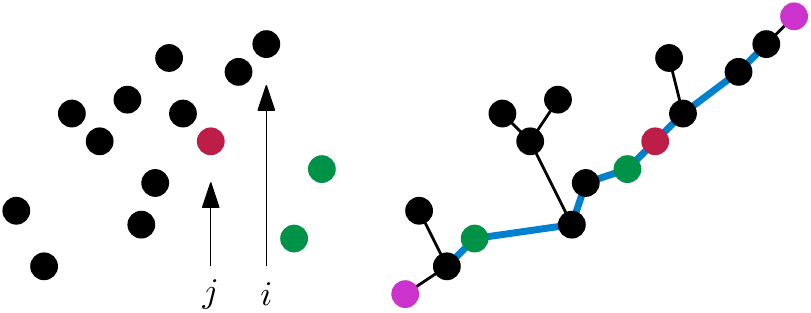}
\caption{On the left we have a representation of $\sigma$, a dot for each 
element $x$ is drawn at the coordinate $(a,b)$ where $x = l[a] = l'[b]$. On the 
right the elements have been moved to their position in $\hat{\sigma}$ and the 
corresponding Cartesian tree is superimposed.  The active element of insertion sort 
at the current moment is highlighted in red, the elements that haven't been seen by 
the algorithm are highlighted in green, the added elements are highlighted in 
pink, and the minima path is highlighted in blue.}\label{fig:cart-trees}
\end{figure}

Recall that $B_t$ is the number of bad inversions, which is the sum of the blocked and stuck inversions. For the purposes of bounding $B_t$, we conceptually associate two counters, $\textit{Inc}(x)$ and $\textit{Dec}(x)$, with each element, $x$. The counters are initialized to zero at the start of an insertion sort round. When an element $x$ is increased by a random swap in $l'$, we increment $\textit{Inc}(x)$ and when $x$ is decreased by a random swap in $l'$, we increment $\textit{Dec}(x)$. After the random swap occurs, we may choose to exchange some of the counters between pairs of elements.

\incdeccounterlemma*

\begin{figure}[hbt]
\centering
\includegraphics{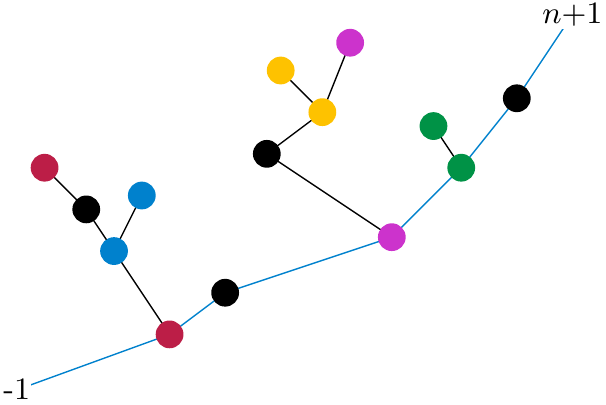}
\caption{Every degree-three vertex is paired up with a leaf in one of it's subtrees. The node $-1$ is always paired with node $n+1$.}\label{fig:pairing-up-examp}
\end{figure}

Maintaining Invariant~1 in the face of the random swaps in $l'$ can
be difficult, because new minima could be added to the path or old
minima could be removed from the path. To handle these challenges,
we pair up each element with degree three in the Cartesian tree with
a descendant leaf. First, as a special case, the $-1$ element in the
Cartesian tree is paired with the $n+1$ element. To find pairs for
the degree-three elements, we consider traversing the tree in depth
first order starting at the root. Below a degree-three element in the
Cartesian tree there are two subtrees. When a degree-three element is
encountered in the traversal, the larger of the maximum leaf element
in the left subtree and the maximum leaf element in the right subtree
will have already been paired up. So we pair the degree-three element
with the unpaired (and smaller) of the two maximum leaves
(Figure~\ref{fig:pairing-up-examp}). For a degree-three element, $l[a]$, denote the index in $l$ of its pair with $P(a)$.
We enforce the following stronger invariant:

\textbf{Invariant 2.} 
\textit{For every element $l[a]$ with degree three in the Cartesian tree,}
$\sigma\bigl(P(a)\bigr) - \sigma(a) \leq \textit{Inc}\bigl(l[P(a)]\bigr) + \textit{Dec}\bigl(l[a]\bigr)$.

Invariant 2 implies Invariant 1, because each minima along the path is either paired with the maximum leaf element in its left subtree if it has one.

We now consider how to maintain Invariant 2 after each random swap in
$l'$. Suppose $\sigma(a) = k+1$ and $\sigma(b) = k$ are the swapped
pair and for now assume neither is the active element. After the swap
$\sigma(a) = k$ and $\sigma(b) = k+1$ and the two counters
$\textit{Dec}\bigl(l[a]\bigr)$ and $\textit{Inc}\bigl(l[b]\bigr)$ are incremented. However, the slight upward and downward movement of elements may have changed how elements are paired up either by a structural change in the Cartesian tree or exchanging the relative value of two leaf elements. There are several cases to analyze based on how the random swap affected the modified Cartesian tree.

First we observe that if the random swap did not affect the pairing of elements, then the incrementing of counters maintains the invariant. For example, if $a$ has a pair $P(a)$, then $\sigma\bigl(P(a)\bigr) - \sigma(a)$ is increased by one and if there is an element $l[c]$ with $P(c) = b$, then $\sigma(b) - \sigma(c)$ increased by one. Each of these increases are offset by the incrementing of $\textit{Dec}\bigl(l[a]\bigr)$ and $\textit{Inc}\bigl(l[b]\bigr)$ respectively.

If the random adjacent swap did affect the pairing of elements, then either $l[a]$ and $l[b]$ are adjacent in the tree or $l[a]$ and $l[b]$ are leaf elements with least common ancestor $l[c]$. In this second case, there is an ancestor of $l[c]$ paired with $l[a]$ before the swap which is paired with $l[b]$ after and $l[c]$ is paired with $l[b]$ before the swap and is paired with $l[a]$ after. For both pairing changes, the distance between the paired elements is unchanged, but the $\textit{Inc}$ counter of the leaf element in the pairs may be incorrect. So we exchange $\textit{Inc}\bigl(l[a]\bigr)$ and $\textit{Inc}\bigl(l[b]\bigr)$.

In the case where $l[b]$ and $l[a]$ are adjacent in the tree, before the swap $l[b]$ is the parent of $l[a]$ and afterwards $l[a]$ is the parent of $l[b]$. When this happens, if either $l[a]$ or $l[b]$ are unsorted elements, then both elements must lie on the minima path and the swap simply exchanges their order on the minima path. So while there is a change in the tree structure, there is no change in the pairing of elements.

We can now assume both elements are semi-sorted which leads to some case analysis based on the degrees of $l[a]$ and $l[b]$ which determines how they are paired with other elements. In these cases, the random swap acts almost like a tree rotation.

\begin{itemize}
\item If $l[a]$ and $l[b]$ both have degree three, then together there are three subtrees below $l[a]$ and $l[b]$. For the largest elements in these three subtrees, one is paired with $l[a]$, one is paired with $l[b]$, and the third is paired with an ancestor of $l[a]$ and $l[b]$. After the random swap, the ancestor will have the same paired element, but $l[a]$ and $l[b]$ may have had their pairs exchanged.
In this case, to maintain our invariant if the pairings changed, we exchange $\textit{Dec}\bigl(l[a]\bigr)$ and $\textit{Dec}\bigl(l[b]\bigr)$. 

This case is shown in Figure~\ref{fig:exchange-examp}.

\item If either $l[a]$ or $l[b]$ has degree three and the other has degree two, then there are two subtrees below $l[a]$ and $l[b]$ in the subtree. Out of the two maximums in the subtrees, one is associated with whichever of $l[a]$ and $l[b]$ has two children and one is associated with an ancestor of $l[a]$ and $l[b]$. Notice that when a swap happens, the degree of $l[a]$ and $l[b]$ will not change if there is a subtree ``between them'' i.e. there are descendants of $l[a]$ and $l[b]$ with index between $a$ and $b$ (or equivalently $|a-b|\neq 1$).

When there is no subtree between $l[a]$ and $l[b]$, then the swap exchanges the degrees of the two elements.
In this case, to maintain the invariant we also exchange $\textit{Dec}\bigl(l[a]\bigr)$ and $\textit{Dec}\bigl(l[b]\bigr)$. 

\item If $l[a]$ has degree one and $l[b]$ has degree three, then
there is only one subtree below $l[a]$ and $l[b]$. Because $\sigma(a)
= \sigma(b) + 1$, that subtree's maximum must be larger than
$\sigma(a)$. So $P(b) = a$. After the swap, this pairing relationship
is destroyed, because both elements will have degree two.
In this case, no additional work is needed to maintain the invariant.

\item If $l[a]$ and $l[b]$ both have degree two, then there is only one subtree below $l[a]$ and $l[b]$. Again we condition on whether or not there is a subtree between $l[a]$ and $l[b]$.

If there is a subtree between them, then the swap simply reorders $l[a]$ and $l[b]$ on the path leading to that subtree causing no change in pairings and maintaining the invariant.

When there is no such subtree, after the swap, one of $l[b]$ will now be a leaf, $l[a]$ will have degree three, and $P(a) = b$.
In this case, a new pairing relationship was created between $l[a]$ and $l[b]$. The swap incremented $\textit{Dec}\bigl(l[a]\bigr)$ and $\textit{Inc}\bigl(l[b]\bigr)$ so $\sigma\bigl(l[a]\bigr) - \sigma\bigl(l[b]\bigr) = 1 <  2 \leq \textit{Inc}\bigl(l[b]\bigr) + \textit{Dec}\bigl(l[a]\bigr)$ and the invariant holds.

\item If $l[a]$ has degree one and $l[b]$ has degree two, then there are no subtrees below $l[a]$ and $l[b]$. After the swap, they will switch which element is the leaf. An ancestor was paired with $l[a]$ and is now paired with $l[b]$. In this case, to maintain the invariant we exchange $\textit{Inc}\bigl(l[a]\bigr)$ and $\textit{Inc}\bigl(l[b]\bigr)$.
\end{itemize}

\begin{figure}[hbt]
\centering
\includegraphics{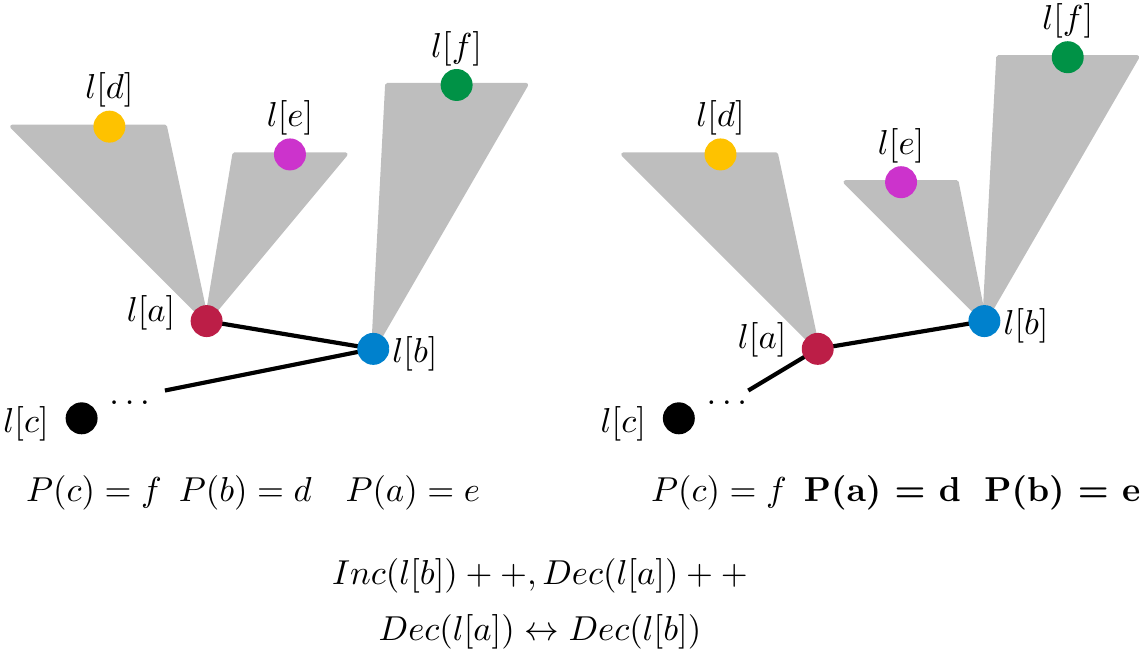}
\caption{When the red and blue element are randomly swapped in $l'$, they switch paired elements and exchanging their $\textit{Dec}$ counters maintains the invariant.}\label{fig:exchange-examp}
\end{figure}

When the random adjacent swap in $l'$ involves the active element,
the affect on the Cartesian tree can be somewhat more complicated.
Issues might arise because $l[j]$ is not yet slotted into its simulated final
horizontal position in the Cartesian tree. We need to make sure the
horizontal movements of the active element do not invalidate the
invariant. Suppose there is a maximal index $k < j$ such that
$\sigma(k) < \sigma(j)$, i.e., index $k+1$ is where the insertion of $l[j]$ will stop. When there is no such $k$, $l[j]$ will be inserted at the front of the list and so we set $k$ to be $-1$. If $l[j]$ swaps with an element outside the range $[k,j - 1]$, then no horizontal movement of $l[j]$ will occur and we can handle the case as though $l[j]$ is semi-sorted.

So suppose $l[j]$ is swapped with $l[a]$ with $a \in [k,j-1]$ and $\sigma(a) = \sigma(j) + 1$ before the swap.  After the swap, $l[j]$ will be moved immediately to the right of $l[a]$ in the Cartesian tree and is the right child of $l[a]$. Because $\sigma(j)$ is smaller than $\sigma(x)$ for $x\in [k,j-1]$, $l[a]$ must be the right child of $l[j]$ before the swap. So $l[j]$ has degree two and is unpaired before the swap.

\begin{itemize}
\item If $l[a]$ had a right child before the swap, then $l[j]$ now subdivides the edge from $l[a]$ to its old right child and has degree two. So the invariant is maintained.
\item If $l[a]$ had only a left child before the swap, then $l[a]$ is
now paired with $l[j]$, which is a leaf after the swap. The invariant
requires $\sigma(j) - \sigma(a) = 1 \leq \textit{Inc}\bigl(l[j]\bigr) + \textit{Dec}\bigl(l[a]\bigr)$. This
inequality is satisfied, because the swap incremented $\textit{Inc}\bigl(l[j]\bigr)$.
\item If $l[a]$ was a leaf paired with $l[c]$ before the swap, then $l[j]$ is now paired with $l[c]$. Exchanging the $\textit{Inc}$ counters for $l[j]$ and $l[a]$ guarantees the invariant is maintained.
\end{itemize}

Now we consider the final case where $l[j]$ is swapped with $l[a]$ with $a \in [k,j-1]$ and $\sigma(a) +1 = \sigma(j)$ before the swap. Because $\sigma(a) < \sigma(j)$, $a=k$. Additionally we observe that $l[j]$ is the right child of $l[a]$ in the Cartesian tree before the swap. After the swap, $l[a]$ is the right child of $l[j]$ and $l[j]$ has degree two. So $l[j]$ is unpaired after the swap.

\begin{itemize}
\item If $l[j]$ had a right child before the swap, then $l[j]$ now subdivides the edge from $l[a]$ to its old parent and has degree two. So the invariant is maintained.
\item If $l[j]$ is a leaf and $l[a]$ has a left child, then $l[a]$ was paired with $l[j]$ before the swap. After the swap, $l[a]$ and $l[j]$ both have degree two with $l[j]$ subdividing the old edge between $l[a]$ and its parent.
\item If $l[j]$ is a leaf and $l[a]$ does not have a left child, then there is some ancestor paired with $l[j]$. The pairing will switch to $l[a]$ after the swap. Exchanging the $\textit{Inc}$ counters for $l[j]$ and $l[a]$ maintains the invariant.
\end{itemize}

\end{appendix}

\end{document}